\documentclass{article}
\usepackage{amsmath,amssymb,amsfonts,amsthm}
\usepackage{color,hyperref}
\usepackage{graphicx}
\usepackage{tikz}

\newcommand{\subscr}[2]{#1_{\textup{#2}}}

\newcommand{\setdef}[2]{\{#1 \, : \; #2\}}
\newcommand{\map}[3]{#1: #2 \rightarrow #3}

\newcommand{\dist}{{\rm dist}}

\newcommand{\realnonnegative}{\mathbb{R}_{\geq0}}

\newcommand{\card}[1]{|#1|} 

\providecommand{\Pr}{\mathbb{P}}
\providecommand{\com[2]}{{\color{blue} \begin{tt}[#1: #2]\end{tt}}}

\newtheorem{definition}{Definition}
\newtheorem{proposition}{Proposition}
\newtheorem{theorem}{Theorem}
\newtheorem{example}{Example}
\newtheorem{remark}{Remark}
\newtheorem{corollary}{Corollary}

\newcommand{\G}{\mathcal{G}}

\title{Protecting shared information in networks: a network security game with strategic attacks}

\author{
Bram de Witte\thanks{B. De Witte, P. Frasca and J.\ Timmer are with Department of Applied Mathematics, University of Twente, 7500 AE Enschede, The Netherlands.
{\tt\small j.b.timmer@utwente.nl}.}
\and 
Paolo Frasca$^{\ast}$\thanks{P.~Frasca is with Univ.\ Grenoble Alpes, CNRS, Inria, Grenoble INP, GIPSA-lab, F-38000 Grenoble, France. {\tt\small paolo.frasca@gipsa-lab.fr}.}
\and Bastiaan Overvest\thanks{B.~Overvest is with CPB Netherlands Bureau for Economic Policy Analysis, The Hague, The Netherlands. {\tt\small b.overvest@cpb.nl}.} 
\and 
Judith Timmer$^{\ast}$%
}

\begin{document}
\maketitle
\begin{abstract}
A digital security breach, by which confidential information is leaked, does not only affect the agent whose system is infiltrated, but is also detrimental to other agents socially connected to the infiltrated system. Although it has been argued that these externalities create incentives to under-invest in security, this presumption is challenged by the possibility of strategic adversaries that attack the least protected agents. In this paper we study a new model of security games in which agents share tokens of sensitive information in a network of contacts. The agents have the opportunity to invest in security to protect against an attack that can be either strategically or randomly targeted. We show that, in the presence of random attack, under-investments always prevail at the Nash equilibrium in comparison with the social optimum. Instead, when the attack is strategic, either under-investments or over-investments are possible, depending on the network topology and on the characteristics of the process of the spreading of information. Actually, agents invest more in security than socially optimal when dependencies among agents are low (which can happen because the information network is sparsely connected or because the probability that information tokens are shared is small). These over-investments pass on to under-investments when information sharing is more likely (and therefore, when the risk brought by the attack is higher).
\end{abstract}


\section{Introduction}

Our society and economy have become largely dependent on sharing information over networks. 
{ Although in general computer networks provide benefits, they are also prone to cyber attacks, whose impact increases with our dependence on them. Security breaches can have various origins, such as the spread of malware, compromissions of social network accounts, or exploitations of system vulnerabilities. In this paper, we interested in cyber attacks where, without permission, confidential information is obtained. This information may represent  for instance confidential documents, intellectual property or identity information. The impact of having sensitive information stolen can be destructive: bank accounts can be plundered, companies can be threatened that strategic decisions or sensitive information will be released or identities can be stolen for criminal purposes. These forms of cyber attacks where confidential information is obtained are occurring more often and keeping personal information out of the hands of thieves is becoming increasingly difficult~\cite{Jang-Jaccard}.

From both the scientific literature and the general media~\cite{Varian00}, it is apparent that the variety of potential threats is huge. Depending on their purpose, some attacks aim at compromising a whole class of systems or devices, whereas others aim at precise targets. We shall refer to the former type of attacks as {\em random} attacks and to the latter type as {\em strategic} attacks~\cite{Acemoglu2016,Lou:2015}.
%
An example of a random cyber attack would be WannaCry. This is a ransomware virus that in 2017 infected about 200,000 computers worldwide, including computers of the National Health Service in the UK, Renault in France and Telefonica in Spain: these computers were all vulnerable because their operators failed to install in time a simple software patch - i.e. arguably under-invested in security measures. Examples of strategic cyber attacks are quite common. A well-known attack is the 2016 security breach against the Democratic National Committee, by which thousands of e-mails were stolen and subsequently leaked, including e-mails from Hillary Clinton. }

{ 
Researchers have soon recognized that network security is not only a matter of devising suitable security measures, but also of making sure that individuals put them into practice~\cite{Varian00}. Consequently, the adoption of security measures has been regarded as an economic problem and has been addressed with the tools of game theory. In this perspective, the key observation is that the presence of a network introduces {\em interdependencies between risks and costs incurred by the individuals}~\cite{Heal03}. Hence, the interesting question becomes understanding the effects of these interdependencies. In order to answer this question, a large literature has been developed not only in the economic science but also in computer science~\cite{Lou:2015} and in engineering, including security problems for wireless communication~\cite{LiKoPo07} and interdependent control systems by~\cite{Amin2013,Zhu15,Yuan16,Gupta17}. These works have addressed an array of questions that are relevant in our own work, including  security games featuring strategic attacks~\cite{CaCeOr12}, multiple targets~\cite{Lou:2015}, multiple attackers and defenders~\cite{ZhTeBa10}.
In this Introduction, we will not trying to provide a complete literature survey on interdependent security, for which we can point the reader to sources like~\cite{Anderson,Manshaei2013,Laszka2014,HeDai18}. Instead, we will more modestly highlight a few recurring issues that motivate our work on interdependencies in network security.}


A number of papers \cite{Anderson,Lelarge09} have argued that security investments are not as high as they should be due to \textit{externalities} in the network. These externalities originate because confidential information can be leaked through other channels than one's own device. As a consequence, agents face risks whose magnitude depend not only on their own security levels but also on the security levels of others. In this setting, investments act like \textit{strategic complements} as benefits of security adoption are not exclusively for the one that invested in the security. Consequently, a negligent agent who does not adequately protect his and others' information due to free-riding, may cause considerable damage to other agents in the network. This leads to situations where benefits of security adoption might fall significantly below the cost of adoption, which causes under-investments. 

More recently, the prediction of under-investments in information networks has been challenged. Acemoglu \textit{et al.\ }\cite{Acemoglu2016} and Bachrach \textit{et al.\ }\cite{Bachrach13} show that investments in security might as well be \textit{strategic substitutes} when agents face an intelligent threat. In their setting, an attacker can aim at the weakest nodes: in this case, a negligent agent who does not invest in security has a relative higher chance that his information is stolen by a direct attack of the hacker. This eliminates the ability to free-ride on security investments of others and forces an agent to invest. In fact, this framework leads to incentives which correspond to an arms race: agents compete with each other leading to over-investments in security. Bachrach \textit{et al.\ }even propose that an optimal policy requires taxing security, contrarily to subsidizing security as recommended by models that do not include an intelligent adversary. 

{ 
Our work provides a tractable model of network security game that can explain both under-investments and over-investments, depending on the strategy of attack and on the amount of shared information,  which eventually depends on the topology of the network connecting the agents. Our original framework and our results can be informally described as follows. 
Inspired by attacks that aim at recovering sensitive information, such as the DNC hack, we define a dissemination model where interconnected agents share confidential information (e-mails€™) with each other with a certain probability $p$, resulting in a dissemination of information among peers that depends on the network structure.} Agents store information (both their own and that received from others) and invest in security to protect it. A malignant and possibly intelligent attacker, who has the goal to obtain as much information as possible, attacks one of the agents. If the attack is successful, the attacker acquires all the information that was stored by the agent, thus making this agent and possibly also other agents, which have entrusted their information to the attacked agent, victims of the attack. If the attacker is able to optimally choose which agent to attack, the attack will be said to be \textit{strategic}: otherwise, to be \textit{random}. 
In our model, the security investments are the outcome of the resulting two-stage game between the agents and the attacker, where the attacker knows the investments of the agents, who in turn choose their investments anticipating the strategy of the attacker. Under this game structure, we show that when the attack is random, then the equilibrium investments are lower than the socially optimal investments. Instead, if the attack is strategic, then the relation between optimal and equilibrium investments depends on the amount of information shared: when the fraction of shared information is low, equilibrium investments are higher than optimal ones, whereas the opposite happens when the fraction of shared information is high.
 { The fraction of shared information can be high for two distinct reasons: either the diffusion probability $p$ is high, or the network is tightly connected. The latter case shows our results to be consistent with those in \cite{CaCeOr12}, where the authors adapt interdependent security games to model strategic attacks and find that€™ over-investments prevail in nondense networks. 

In order to keep our analysis tractable, some of our results on strategic attacks make an assumption of homogeneity in the network, namely that the network is vertex-transitive. We complement these results with an analysis of the security game on star graphs, which we choose as a natural example of non-homogeneous topology: this case study shows that the essential lines of our findings, as we described them above, remain valid on general graphs.}

This paper is structured as follows. 
Section~\ref{sec:model} describes the problem that we want to address, introducing the dissemination model, the attack and the security investments.
Subsequently, Section~\ref{sec:spreading} examines the dissemination model that underlies the security game. 
Sections~\ref{sec:random} and~\ref{sec:security} are the core of our paper, as they study the security game when the attack is random and when the attack is strategic, respectively. Finally, Section~\ref{sec:discussion} discusses { the obtained results and complements them with numerical evaluations on complete, ring, and star graphs that we have selected as fundamental examples. Section~\ref{sec:outro} summarizes and concludes the paper. The body of the paper is complemented by an Appendix which is devoted to detailed derivations (for the examples of complete and star graphs) and proofs (of the main results about strategic attacks).}


\section{Information dissemination and network game}\label{sec:model}
Our dynamics of interest take place on a network of agents that can share tokens of information, such as confidential documents, with each other. Let us think of $n$ agents in a set $V=\{1,\dots,n\}$. We say that two agents $i$ and $j$ are linked by an edge $(i,j)$ when $i$ and $j$ can share documents directly with each other. These edges create a (undirected) network $\mathcal{G}= \langle V, A \rangle$, where $A: V \times V \rightarrow \{0,1\}$ is the adjacency matrix in which  $A(i,j) = A(j,i) = 1$ if and only if $i$ and $j$ are linked by an edge. We denote the set of all edges in $\mathcal{G}$ as $E(\mathcal{G})$. In this graph theoretical context, we need to recall some standard definitions.
A path $u$ in $\mathcal{G}$ between agent $i$ and $j$ is a sequence of distinct edges $u=\{(i,\kappa_1),(\kappa_1,\kappa_2), \dots, (\kappa_{\ell-1},\kappa_\ell),(\kappa_\ell,j)\}$, where $|u|=\ell$ is the length of the path. We assume that $\mathcal{G}$ is a network in which there exists a path between all pairs of agents, in other words, $\mathcal{G}$ is a connected network. A subnetwork $\mathcal{G}'=(V',A')$ of $\mathcal{G}$ is a network such that $V'\subset V$ and $E(\G')\subset E(\G)$. In Figure~\ref{fig:Interesting} we illustrate these concepts and show some networks of interest.
\begin{figure}
\centering
\includegraphics[width=1\linewidth]{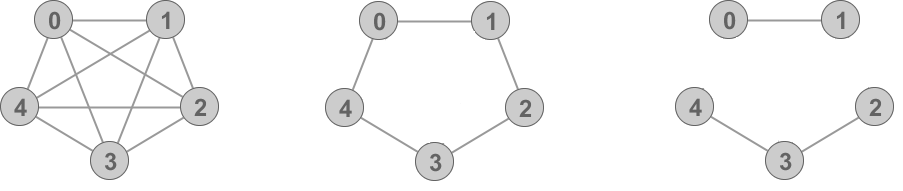}
\caption{The leftmost network is a complete network and the middle one is a ring network. In the ring, $\{(0,1)(1,2)(2,3)\}$ is a possible path from agent 0 to agent~3. As each edge in the ring is also in the complete network, the ring is a subnetwork of the complete network. While the rightmost network is not connected, it is a subnetwork of the ring and of the complete network.} 
\label{fig:Interesting}
\end{figure}

Our problem statement requires us to specify three key ingredients: (i) the dissemination of information, (ii) the adversary attack, (iii) the defensive investments.

\paragraph{Information dissemination model.}

We assume that initially every agent owns a unique document which we will denote as $d_i$ for agent $i$. All the $n$ documents spread, independently of each other, over the network $\mathcal{G}$. Although the documents are confidential, it is not detrimental for an agent when his document is obtained by other agents. We assume that an agent obtains a document from another agent with probability $p$ when they are connected. This leads to a so-called \textit{transmission network} for each document. A generic transmission network $\mathcal{T}$ is a random subnetwork of $\mathcal{G}$ and formally defined as $\langle V, \tilde{A} \rangle$, where
\begin{align*}
\tilde{A}(i,j) = \tilde{A}(j,i) = X_{ij}~ A(i,j)
\end{align*}
where $X_{ij}$ are independent random variables identically distributed according to a Bernoulli distribution with parameter $p$. We are thus assuming that the probability of transmission between two neighboring nodes is identical for every document.  
Let $\mathcal{T}_\ell$ and $x_{ij,\ell}$ be instances of transmission networks and transmission probabilities for a dissemination starting from any $\ell\in V$. Then $\mathcal{T}_\ell=\langle V, \tilde{A}_\ell \rangle$ with $\tilde{A}_\ell(i,j) = \tilde{A}_\ell(j,i) = X_{ij,\ell}~ A(i,j)$. Now, an agent obtains document $d_\ell$ when she is connected to agent $\ell$ in the {transmission network} $\mathcal{T}_\ell$.
The spread of the $n$ documents then is described by the  $n$ transmission networks.

The network structure determines the probability that a document spreads from its owner to another agent. We define the matrix $P$ with elements $P_{ij}$ representing the probability that agent $j$ owns document $d_i$ after dissemination
\begin{align} \label{eq:P}
P_{ij} &=\Pr\{\text{there exists a path between $i$ and $j$ in $\mathcal{T}_i$}\} \\
\label{eq:P-paths}
& = \Pr \{ \bigcup_{u\in U_{i,j}(\mathcal{G}) } \{ u \in U_{i,j}(\mathcal{T}_i)\}  \},
\end{align}
where $U_{i,j}(\G)$ is the set of all paths between agent $i$ and $j$ in network $\mathcal{G}$. Note that the  matrix $P$ is symmetric and only depends on $\mathcal{G}$ and on $p$.
Since we assumed that $\mathcal G$ is connected, $P$ contains only strictly positive elements. Although $P_{ij}=P_{ji}$, the event that $j$ obtains $d_i$ is independent of $i$ obtaining $d_j$, because they are respectively taking place on the transmission networks $\mathcal{T}_i$ and $\mathcal{T}_j$.
 Denote the expected number of documents obtained by agent $i$ as $D_i $ and note that
\begin{equation}\label{eq:D}
D_i  =\sum_{j \in V} P_{ji} = \sum_{j \neq i} P_{ji} + 1.
\end{equation}
We additionally denote $\mathbf{D} = \{D_1,\dots,D_n\}$. 

\paragraph{Attack model.}
After the documents have spread through the network, the adversary attacks {\em one} agent. We model this attack by a random variable from a distribution over the agents. This distribution is conveniently represented by the probability vector 
$\mathbf{a} = \{a_1,\dots,a_n\}$ which we call \textit{the attack vector}. When an attack on an agent is successful the attacker will steal all the documents stored at the target. This always includes an agent's own document, but may additionally include documents of other agents. We assume that the attack vector is established before the documents spread through the network. 

\paragraph{Defense model.}
Before the attack vector is chosen, agents have the opportunity to precautionary invest in security. We denote these investments $\mathbf{q}=\{q_1,\dots,q_n\}$ as \textit{the security vector}. These security investments are such that an attack on agent $i$ is successful with probability $1-q_i$. Let $x_i=1$ denote the event that {\em the attacker obtains document $d_i$}, and $x_i=0$ otherwise. Consequently, by conditioning and exploiting independence we establish that
\begin{align}\label{eq:probdamage}
\Pr\{x_i = 1 \} \hspace{1mm}
& = \sum_{j \in V } \hspace{1mm} a_j (1-q_j) P_{ij}.
\end{align}
Recognize that the security of an agent $i$, that is, the privacy of his information $d_i$, does not only depend on his own investment, but also depends on the investments by the other agents.
Furthermore, let $|\mathbf{x}|=\sum_{i\in V} x_i$. Observe that the expected number of stolen documents is 
\begin{align}\label{eq:totaldamage}
	\mathbb{E}(|\mathbf{x}|) = &\sum_{i \in V}\Pr\{x_i = 1\} \nonumber\\
	= &\sum_{i \in V} \sum_{j \in V} a_j (1-q_j) P_{ij} \nonumber\\
	=& \sum_{j \in V} a_j (1-q_j) D_j,
\end{align}
because the attacker affects only one node directly. 
%

\paragraph{Problem summary.}
The timing in our problem is as follows. Firstly, the agents invest in security by selecting the security vector $\mathbf{q}$. Secondly, the attacker chooses the attack vector $\mathbf{a}$, possibly in order to maximize his reward. Hereafter, the documents spread through the network. Finally, one agent is attacked by the attacker. 
Since in our model the attacker observes the security levels of all the agents, the relevant equilibrium concept is that of the Stackelberg equilibrium of the resulting two-stage game \cite{peters2016game}: the agents first select their security levels anticipating the decision of the attacker (as they know his strategy) and the attacker optimizes his { attack strategy while having knowledge of the security choices. 
Let us note that the proposed sequence of players actions in the Stackelberg game (first defender, then attacker) is the interesting one to study. On the contrary, the reverse sequence would be unrealistic and would trivialize the game. Should the attacker go first, then the agents would just know who is to be attacked and would simply be able to optimally protect the target node. 
}

\paragraph{More on the relation with literature on contagion and security games.} { 
The two-stage scheme of the strategic security game that we study in this paper is adopted from the work \cite{Acemoglu2016} on cascading failures and contagion. However, our problem statement is different because the underlying diffusion/contagion model is different. 
In~\cite{Acemoglu2016}, each agents is susceptible to the attack with probability $1-q_i$ and the infection spreads from the attacked node to all nodes that are connected to it in the sub-network spanned by the susceptible nodes. Therefore, an investment in security prevents both contagion from a direct attack and contagion through the network. Instead, in our model the susceptibility is only realized at the attacked node, whereas the dissemination of information independently takes place across all edges for all pieces of information. Therefore, nodes cannot be safe from damage even if they invest maximally in security, since their private information is shared with other nodes. Another difference is the explicit presence of the variable $p$, the probability of diffusion: in our results the amount of over- or under-investments is dependent on the level of interdependence in the network, which is directly influenced by the network topology and by the probability $p$. }

\section{Information dissemination}\label{sec:spreading}

The next proposition provides more insight about the value of $D_i$, the expected number of documents obtained by agent $i$. Its proof is straightforward and therefore omitted. In order to emphasize the dependence of $D_i$ on $p$ and $\mathcal{G},$ we shall use the notation $D_i(p,\G)$. The result is illustrated in Figure~\ref{DRingComplete}.

\begin{proposition}[Monotonicities]\label{prop:Dbounds}
Given a network $\mathcal{G}$,  $D_i(p,\G)$ is strictly increasing in $p$ for all $i$.  Given two networks $\mathcal{H}\subset \mathcal{G}$,  $D_i(p,\mathcal{H})\le D_i(p,\G)$, provided node $i$ belongs to both networks.
\end{proposition}

\begin{example}[Star graph]\label{exa:star}
Consider a star graph with $n$ nodes: node $1$ is the center and the remaining $n-1$ nodes are the leaves. Note that (with $i,j>1$ and $i\neq j$) $$P_{1i}=p\qquad P_{i1}=p \qquad P_{ij}=p^2.$$
Hence,
$$D_1=(n-1)p+1 \qquad D_i=(n-2)p^2+ 1+p,$$
which implies that $D_1>D_i$.
\end{example}

\bigskip
In order to make our analysis tractable, we will often assume the networks to be vertex-transitive. Although this choice limits the scope of our results, we conjecture that economic forces in vertex-transitive networks extend to a broader class of networks.
Informally, a vertex-transitive network is a network which `looks the same' at every node. More precisely, we adopt the following definition.
\begin{definition}[Vertex transitivity]
A network $\mathcal{G}$ is vertex-transitive if and only if for any two nodes $i$ and $j$ there exists a mapping $\phi$ such that $\phi(i)=j$ while the structure of $\mathcal{G}$ is preserved: $A(\kappa_1,\kappa_2)=A(\phi(\kappa_1),\phi(\kappa_2))$ for all $\kappa_1,\kappa_2 \in V$.
\end{definition}
The two leftmost networks in figure \ref{fig:VTnetworks} are examples of vertex-transitive networks. While every agent in a vertex-transitive network has the same number of other agents whom she is linked to (regular network), the converse is not necessarily true. As an example, the last network in figure \ref{fig:VTnetworks} is regular but not vertex-transitive.  
\begin{figure}[h]
\centering
\includegraphics[width=1\linewidth]{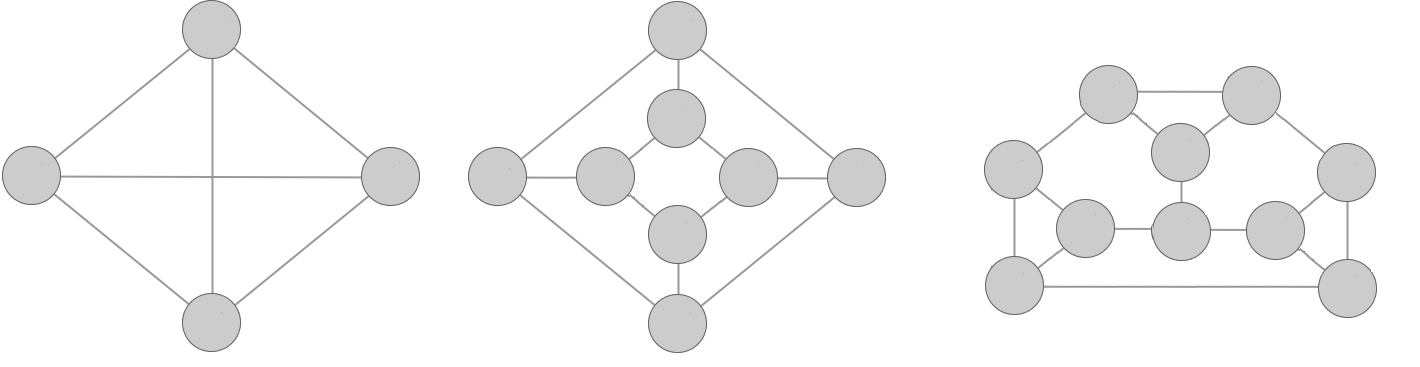}
\caption{Several 3-regular networks. The complete network with 4 agents and the middle network are vertex-transitive networks. The last network is an example of a network which is regular but not vertex-transitive.} 
\label{fig:VTnetworks}
\end{figure}
It is no surprise that every agent in a vertex-transitive networks obtains --- in expectation --- the same number of documents. We state this formally in the next proposition.

\begin{proposition}[Shared documents in vertex-transitive networks]
In any vertex-transitive network, $D_i = D_j$ for all $i,j \in V$. 
\end{proposition}
\begin{proof}
By vertex-transitivity there exists a $\phi$ such that $\phi(i)=j$ while the structure is preserved, which means that $P_{\ell k}=P_{\phi(\ell) \phi(k)}$. Consequently by (\ref{eq:D}),
$
D_i = \sum_{k \in V} P_{k,i} = \sum_{\phi(k) \in V} P_{\phi(k),j} = D_j,
$
yielding the result.
\end{proof}

Since on vertex-transitive networks all elements in $\mathbf{D}$ are identical (for all values of $p$), we will adopt the notation $D_i=D$. 
Complete graphs and ring graphs are both examples of vertex-transitive networks.

\begin{example}[Ring graph]\label{exa:ring}
Consider a ring graph with $n$ nodes (see Figure~\ref{fig:Interesting}). Let $\dist(i,j)=\min\{|i-j|,n-|i-j|\}$ be the distance between nodes $i$ and $j$. By a simple inclusion-exclusion reasoning, observe that if $j\neq i$ then $$P_{ij}=p^{\dist(i,j)}+p^{n-\dist(i,j)}-p^n.$$
Hence, by summing over the nodes 
$$D=1+2\sum_{\ell=1}^{n-1} p^\ell - (n-1)p^n=\frac{1+p - p^n (n+1) + p^{n+1} (n-1)}{1-p}.$$
Note that $D\to \frac{1+p}{1-p}$ as $n\to \infty$. In contrast, recall from Example~\ref{exa:star} that $D_i$ is unbounded in $n$ on star graphs.
\end{example}

Ring and star graphs are simple to deal with because the number of possible paths between two nodes is small. On the contrary, the complete graph has a very large number of possible connecting paths. Nevertheless, some quantities can be explicitly computed.

\begin{example}[Complete graph]\label{exa:complete}
For the sake of clarity, we denote by $D^{n}$ and $P_{ij}^n$ the expected number of documents and the generic transmission probability on the complete graph $K_n$, respectively. Due to transitivity, 
$$ D^n=1+(n-1)P_{ij}^n$$
and for small $n$ we easily see that $P_{ij}^2=p$ and $P_{ij}^3=p+p^2-p^3$. 
To obtain some more general expressions, let $Q^n$ denote the probability that any document reaches all nodes in $K_n$. Then, $Q^1=1$ and 
\begin{equation}\label{eq:Qk}
Q^k=1- \sum_{\ell=1}^{k-1} {{k-1} \choose {\ell-1}} (1-p)^{\ell (k-\ell)} Q^\ell.
\end{equation}
In turn, 
\begin{equation}\label{eq:Pk}
P_{ij}^n= \sum_{k=2}^{n} {n-2 \choose k-2} (1-p)^{k (n-k)} Q^k.
\end{equation}
These formulas, proved in the Appendix, allow for the numerical evaluation of $D$ on graphs of moderate size, 
{ as shown in Figure~\ref{DRingComplete}.
For large $n$, it is useful to consider the bounds
$$1- (1-p) (1-p^2)^{n-2}\le P_{ij}^n\le 1- (1-p)^{n-1}.$$ 
The lower bound can be obtained by considering only propagation across paths of length at most two. The upper bound can be obtained by considering that the document from $i$ cannot reach vertex $j$ unless at least one edge reaches $j$ in graph $\mathcal{T}_i$.
These two bounds together imply that  $D_i$ is asymptotically linear in $n$.
}
\end{example}
\begin{figure}
\centering
  \includegraphics[width=1\linewidth]{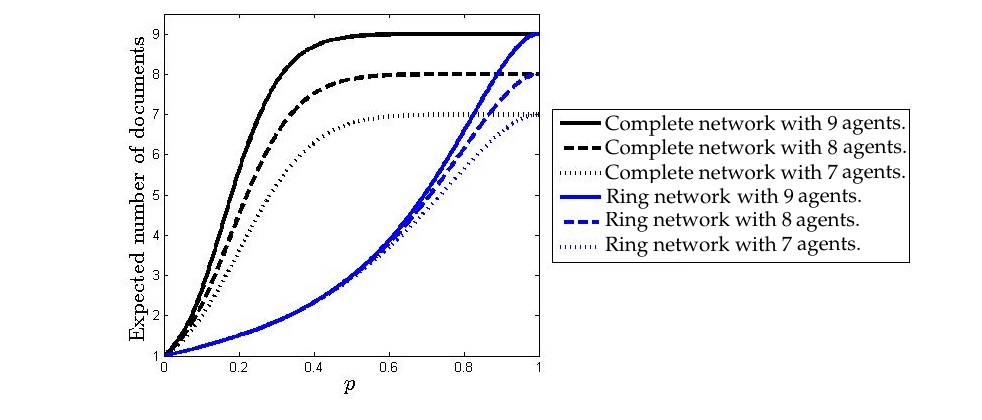}
  \caption{Computations on ring and complete graphs illustrate that the expected number of documents $D$ obtained by each agent is increasing in the density of the network and in $p$. 
  Note that for any graph for which the ring is a subgraph, every $D_{i}$ must be higher than $D$ in the ring and lower than $D$ in the complete graph.}
  \label{DRingComplete}
\end{figure}

\section{Security under random attacks}\label{sec:random}
Security investments are conveniently modeled as the outcome of a game between agents.  
In this section, we look at the social optimum and the equilibria of this security game in the presence of a random attack. The game with a strategic attack is considered in Section~\ref{sec:security}.

The security game with random attacks is defined as follows. A random attack is defined by the uniform attack vector $$a_i=\frac1n \quad \forall i,$$ which is known to all agents.The player set is the set of agents or nodes $V$. The strategy set of agent $i$ is $Q_i=[0,1]$. The reward of each agent $i$ is defined { as the probability that her own document is safe minus the incurred cost, that is,}
\begin{align}\label{eq:ut}
\Pi_i = 1 - \Pr\{x_i = 1\} - c(q_i),
\end{align}
where $\Pr\{x_i = 1\}$ is given in (\ref{eq:probdamage}) and $c(q_i)$ is the cost agent $i$ incurs for choosing $q_i$. We assume that
\begin{align*}
c(q) = \frac{1}{2}\alpha q^2
\end{align*}
for some $\alpha \geq 1$. The choice of a quadratic cost is made for simplicity: the analysis can be extended to other smooth convex increasing functions. The choice of $\alpha$, instead, is meant to make the cost ``large'', so to rule out trivial game outcomes with maximal investments. Also this assumption can be relaxed at the price of more involved analysis.

In this setting each agent attempts to maximize his/her reward while disregarding the utilities of the others. This is described by a \emph{noncooperative game} 
$(V,\{Q_i\}_{i\in V}, \{\Pi_i\}_{i\in V})$ 
with player set $V$. Any player $i\in V$ has strategy set $Q_i$ and payoff function $\Pi_i$. For these games, the classical definition of {Nash equilibrium} is of interest: an investment level $\mathbf{q}^N$ is a pure strategy Nash equilibrium if for any player $i$ and any investment level $q_i \in [0,1]$ unilateral deviation does not pay,
\begin{align*}
\Pi_i(\{q^N_i, \mathbf{q}^N_{-i}\}) \geq \Pi_i(\{q_i, \mathbf{q}^N_{-i}\}).
\end{align*}
Here $(\{q_i, \mathbf{q}^N_{-i}\})$ denotes the vector $\mathbf{q}^N$ where component $i$ is replaced by $q_i$. In security games under random attack, the Nash equilibrium has a simple structure.

\begin{theorem}[Equilibrium against random attack] 
In a security game facing a random attack, the {  pure strategy} Nash equilibrium $\mathbf{q}^{N,R}$ is unique and is equal to
\begin{align} 
q^{N,R}_i = \frac{1}{\alpha n} \quad \forall i.
\end{align}
\end{theorem}
\begin{proof}
The utility of agent $i$ reads $\Pi_i = 1 - \frac1n \sum_j (1-q_j) P_{ij}- \frac12 \alpha q_i^2.$
We easily see that $\frac{\partial \Pi_i}{\partial q_i}=\frac1n-\alpha q_i$ and $\frac{\partial^2 \Pi_i}{\partial q_i^2}=-\alpha<0$. Since $\frac{\partial \Pi_i}{\partial q_i}(0,\mathbf{q}_{-i})>0$ and $\frac{\partial \Pi_i}{\partial q_i}(1,\mathbf{q}_{-i})<0$, we conclude that the largest utility is obtained with investment $q^{N,R}_i = \frac{1}{\alpha n}$, the unique only Nash equilibrium.
\end{proof}

Some remarks are in order. Firstly, the Nash equilibrium does not depend on $p$ or on the network. The economic motivation for this result is intuitive. As an agent cannot control a possible external loss in a random attack, an increase in investments does not lead to a reduced risk that his document is stolen through another agent. This forces an agent --- in a non-cooperative setting --- to ignore the external risk and to find the optimal trade-off between investment costs and protection against a direct loss. 
Secondly, the investment levels at the Nash equilibrium go to zero as the number of nodes goes to infinity. This is because the risk of being attacked is diluted in large networks.

\smallskip
{ In contrast with the above non-cooperative setting}, we may consider a cooperative setting where all agents cooperate to maximize the social utility, which equals the sum of the agents' utilities:
\begin{align}\label{eq:social}
S(\mathbf{q}) &= \sum_{i \in V} \Pi_i = n-\mathbb{E}(|\mathbf{x}|) - \sum_{i\in V}c(q_i).
\end{align}
By the continuity of $S$ on its compact domain $[0,1]^n$, the function $S$ must attain a maximum. That maximum is said to be the social optimum.

\begin{theorem}[Social optimum against random attack]\label{}
In a network facing a random attack, the social optimum $\mathbf{q}^{O,R}$ is unique and is equal to
\begin{align}\label{X6formule3s}
q^{O,R}_i = \frac{D_i}{\alpha n} \quad \forall i.
\end{align}
\end{theorem}
\begin{proof}
By \eqref{eq:totaldamage} the global utility reads $S(\mathbf{q}) = n - \frac1n \sum_j (1-q_j) D_{j}- \frac\alpha2\sum_j  q_j^2.$
We easily see that $\frac{\partial S}{\partial q_i}=\frac1n D_i-\alpha q_i$ and 
$$\frac{\partial^2  S}{\partial q_i^2}=-\alpha<0 \qquad \frac{\partial^2  S}{\partial q_i\partial q_j}=0,$$ implying that $S$ is a concave function of $\mathbf{q}$. Since $\frac{\partial S}{\partial q_i}(0,\mathbf{q}_{-i})>0$ and $\frac{\partial S}{\partial q_i}(1,\mathbf{q}_{-i})<0$ because $D_i\leq n$ and $\alpha\geq 1$, we conclude that $\mathbf{q}^{O,R}$ with $q^{O,R}_i = \frac{D_i}{\alpha n}$ is the unique maximizer.
\end{proof}

{ 
Comparing these results shows that the Nash equilibrium features {\em under-investments} relative to the social optimum.
This happens because in the cooperative setting an agent also invests to protect documents of others. This additional effort leads to higher investments in security, which depend on the network and the probability $p$. 
The following examples illustrate these observations.}
\begin{example}[Ring network, cont'd]\label{exa:ring2}
Consider the ring network studied in Example~\ref{exa:ring} and assume $\alpha=1$. Then,
the socially optimal investments are 
$$q^{O,R}_i = \frac{1+p-p^n (n+1) + p^{n+1}(n-1)}{(1-p)n}\quad i\in V,
$$
and the Nash equilibrium investments remain $q^{N,R}_i=1/n$. Both these quantities decrease to zero as $n$ goes to infinity.
\end{example}

\begin{example}[Star network, cont'd]
Consider the star network studied in Example~\ref{exa:star} and assume $\alpha=1$. Then,
the socially optimal investments are 
$$q^{O,R}_i = \begin{cases}\frac{(n-1)p + 1}{ n} &\quad i=1\\
\frac{(n-2) p^2+ p + 1}{ n} &\quad i>1
\end{cases}$$
Observe that all investments are non-vanishing for $n\to\infty$ and that the central node 1 supports the highest investment. On the contrary, the Nash equilibrium investments $q^{N,R}_i=1/n$ go to zero for $n\to\infty$.
\end{example}

{ 
\begin{example}[Complete network, cont'd]\label{exa:complete2}
Consider the complete network studied in Example~\ref{exa:complete} and assume $\alpha=1$. Then,
the socially optimal investments converge (exponentially fast in $n$) to the maximum investment:  
$$q^{O,R}_i \to 1 \quad\text{as $n\to\infty$},$$
while the equilibrium investments go to zero for $n\to\infty$.
\end{example}

Optimal investments are larger on networks that are well connected. Indeed, they are vanishing in the limit for large $n$ for the cycle graph, which is poorly connected, whereas investments are non-vanishing in the limit for large $n$ on well-connected networks such as stars and complete graphs. Consistently, the complete graph requires the highest security investments.
}

\section{Security under strategic attacks}\label{sec:security}

In the previous section we analysed the security game in the presence of a random attack. As of this section we allow for a strategic attack by the adversary. As such, the adversary and agents are involved in a two-stage game, the so-called \emph{Stackelberg game} \cite{peters2016game}. In the first stage, the agents determine their investments in security. Thereafter, in the second stage, the adversary selects an attack strategy. Such a game is solved by a {Stackelberg equilibrium}.

\subsection{Definition of strategic attack}

We start the analysis with the strategy of the attacker. 
The vector $\mathbf{a}$ is chosen by the attacker in an optimal way, based on the knowledge of the network and of the vector $\mathbf{q}$.
More precisely, we assume that the strategy of the attacker is an optimal trade-off between the expected number of stolen documents and the cost of this attack, solving the following optimization problem
\begin{align}\label{eq:program}
	 \max_{\mathbf a} \quad & \mathbb{E}(|\mathbf{x}|) - \sum_{i\in V}\psi(a_i) \\ \nonumber
	 \text{subject to}\quad & |\mathbf{a}|=1 \text{ and } a_i \geq 0 \text{ for all }i \in V. \nonumber
\end{align}
Here the expected number of stolen documents is $\mathbb{E}(|\mathbf{x}|)$ and the function $\map{\psi}{[0,1]}{\realnonnegative}$ defines the cost the attacker incurs for choosing $\mathbf{a}$. 
Note that this framework is consistent with the attacker playing the Stackelberg game after the defending agents. 
%
%
In this paper we assume quadratic costs
\begin{align*}
\psi(a) = \frac{1}{2}\omega a^2
\end{align*}
with $\omega \geq 1$.  
Note that this definition implies that a more precise attack is more costly than a more random one. Similarly to what was discussed for the agent's cost $c$, extensions to other convex increasing functions are possible.
By using the expression for $\mathbb{E}(|\mathbf{x}|)$ in~\eqref{eq:totaldamage}, the problem becomes
\begin{align}\label{eq:program-quadratic}
	\max_{\mathbf a}\quad   & \sum_{i=1}^n \left(a_i (1-q_i) D_i -\frac12 \omega a_i^2\right)\\
	\text{subject to}\quad & |\mathbf{a}|=1 \text{ and } a_i \geq 0 \text{ for all }i \in V. \nonumber
\end{align}

The Karush-Kuhn-Tucker (KKT) conditions can be used to solve~\eqref{eq:program-quadratic}. As the objective function is strictly concave, these  conditions are necessary and sufficient to obtain the optimal solution. The KKT conditions read
\begin{subequations}\label{X4formule3}
\begin{align}
\label{X4formule3-1}& (1-q_i)D_i - \omega a_i + \lambda+ \kappa_i=0, \qquad \forall i,\\
\label{X4formule3-4}& \sum_{i \in V} a_i = 1,\\
\label{X4formule3-3}& a_i \geq 0, \qquad \forall i,\\ 
\label{X4formule3-5}& \kappa_i\geq 0, \qquad \forall i,\\
\label{X4formule3-2}& \kappa_i a_i = 0,\qquad \forall i,
\end{align}
\end{subequations}
where $\lambda \in \mathbb{R}$ and $\kappa_i \in \mathbb{R}^+$ for all $i$ are the Lagrange multipliers corresponding to the constraints \eqref{X4formule3-4} and \eqref{X4formule3-3} respectively. Solving these conditions results in the following characterization of the optimal attack strategy.

\begin{proposition}[Optimal attack vector]
The optimal attack vector $\mathbf{a}^*$ chosen by the attacker, solving~\eqref{eq:program-quadratic}, is given by the unique solution $(\lambda^*,\mathbf a^*)$ to the equations
\begin{subequations}\label{eq:find-attack}
\begin{align}
\omega =& \sum_{i \in V} \max\{0,(1-q_i)D_i + \lambda\}, \label{eq:findlambda}\\
a_i =&\frac{1}{\omega}\max\{0,(1-q_i)D_i + \lambda\}\quad \forall i\in V. \label{eq:finda}
\end{align}
\end{subequations}
Consequently, $\mathbf a^*$ is a function of $\mathbf q$ and $\mathbf D$ (and in turn of $p$ and of the topology of the network).
\end{proposition}
\begin{proof}
By substituting \eqref{X4formule3-1} into \eqref{X4formule3-4} and noting that by \eqref{X4formule3-2} $\kappa_i=0$ if $a_i>0$, the multiplier $\lambda^*$ must solve
\begin{align*}
\omega = \sum_{i \in V} \max\{0,(1-q_i)D_i + \lambda\} 
\end{align*}
To show that $\lambda^*$ is unique, suppose that there are two solutions of (\ref{eq:findlambda}): $\lambda_1$ and $\lambda_2$. Without loss of generality, assume that $\lambda_1<\lambda_2$ and set $V_k = \{ i \in V~|~(1-q_i)D_i + \lambda_k > 0 \}$ for $k=1,2$. Obviously, $V_1 \subseteq V_2$. Also note that
\begin{align*}
0 = \omega - \omega &= \sum_{i \in V_1} \big((1-q_i)D_i + \lambda_1\big) - \sum_{i \in V_2} \big((1-q_i)D_i + \lambda_2\big)
\\&= -\sum_{i \in V_2\setminus V_1} (1-q_i)D_i + \lambda_1\card{V_1} -\lambda_2\card{V_2}< 0,
\end{align*}
which gives us a contradiction. So, $\lambda^*$ is unique.
Next, \eqref{X4formule3-1} directly leads to~\eqref{eq:finda} and $a_i(\mathbf{q})$ is a well-defined function of $q$ by the uniqueness of $\lambda^*$.
\end{proof}

The example below illustrates the optimal strategic attack probabilities for star networks.
\begin{example}[Star network, cont'd]\label{exa:a-star} 
Consider the star network studied in Example~\ref{exa:star} and, by symmetry, assume that $q_2=\ldots=q_n$. 
We begin by looking for solutions to~\eqref{eq:findattack} such that $a^*_i>0$ for all $i$.
In this case, equations~\eqref{eq:findlambda} and~\eqref{eq:finda} become 
\begin{align*}
1=\omega&=(1-q_1) D_1 + (n-1) (1-q_2) D_2 + n \lambda^*, \\
a^*_1=\omega a^*_1&= (1-q_1) D_1+ \lambda^*, \\
a^*_2=\omega a^*_2&= (1-q_2) D_2 + \lambda^*,
\end{align*}
and $a^*_k=a^*_2$ for $k=3,\ldots,n$.
Solving the first equation for $\lambda^*$ and substituting that in the other two equations yields
\begin{align*}
a_1^*&= \frac{1}{n} + \frac1\omega(1-\frac{1}{n}) \left( (1-q_1) D_1 - (1-q_2) D_2 \right),
\\
a_2^*&= \frac{1}{n} { +} \frac1\omega  \frac1n \left( (1-q_2) D_2  - (1-q_1) D_1 \right).
\end{align*}

Let $\Delta=(1-q_1) D_1 - (1-q_2) D_2$ and observe that $\Delta$ can be either negative or positive and its magnitude is approximately linear in $n$. Since $a^*_1-a^*_2=\frac1\omega \Delta$,  we observe that $a^*_1>a^*_2$ when 
$$ \frac{1-q_1}{1-q_2}>\frac{D_2}{D_1}.$$
Let us refer to $1-q_i$ as the ``risk'' taken by agent $i$. Since $D_2/D_1\to p$ when $n\to \infty$, we may say that on large star networks, the center is more likely to be attacked than the leaves when the center takes more than $p$ times the risk taken by the leaves. 

When $n$ grows larger, the gap between the two attack probabilities increases, until $a^*_1=1$ and $a^*_2=0$ (if $\Delta>0$) or until $a^*_1=0$ and $a^*_2=\frac1{n-1}$ (if $\Delta<0$). 
The former vector is indeed the optimal solution when $ \omega \le \Delta$, whereas the latter is optimal when $ \omega \le -(n-1) \Delta.$

\end{example}


{ 
Since finding explicit solutions to~\eqref{eq:find-attack} quickly becomes unfeasible on more complex graphs, we instead set to investigate qualitative properties of the optimal attack vector. In the next result, we derive how the optimal attack probabilities depend on the investment levels $\bf q$. The formulas confirm the intuition that the optimal attack probability $a^*_i$ is decreasing in the investments $q_i$ of agent $i$, and increasing in the investments $q_j$ of agents $j\neq i$.
}

\begin{proposition}[How attacks depend on investments]\label{prop:marginal-changes}
The marginal changes of the optimal attack probability $a_i^*>0$ to $q_i$ and to $q_j$ for agent $j$ with $a^*_j>0$, are respectively given by
\begin{align}\label{eq:derivatives}
\frac{\partial a^*_i}{\partial q_i} = -\frac{n^*-1}{\omega n^*}D_i ~~\text{ and }~~ \frac{\partial a^*_i}{\partial q_j} = \frac{1}{\omega n^*}D_j
\end{align}
where $n^*=\card{\setdef{i\in V}{a^*_i>0}}$ is the number of agents with strict positive probability of being attacked. In particular, $a_i^*$ is nonincreasing in $q_i$ and nondecreasing in $q_j$.
\end{proposition}
\begin{proof}
The marginal changes follow from the KKT-conditions in \eqref{X4formule3}. First note that $\kappa_i = 0$ when $a^*_i >0$. Consequently when we differentiate KKT-condition \eqref{X4formule3-1} with respect to $q_i$ we get 
\begin{align}\label{eq:marginal1}
\nonumber -D_i -\omega \frac{\partial a^*_i}{\partial q_i} + \frac{\partial \lambda}{\partial q_i} = 0\\
\frac{\partial a^*_i}{\partial q_i}=-\frac{D_i}{\omega}+\frac1\omega\frac{\partial \lambda}{\partial q_i}
\end{align}
and --- similarly --- when we differentiate with respect to $q_j$
\begin{align}\label{eq:marginal2}
\nonumber-\omega \frac{\partial a^*_i}{\partial q_j}+  \frac{\partial \lambda}{\partial q_j} = 0\\
\frac{\partial a^*_i}{\partial q_j}=\frac1\omega\frac{\partial \lambda}{\partial q_j}
\end{align}
Next we combine  KKT-condition \eqref{X4formule3-4} with the observations above. First recognize that the equation $\sum_{j} a^*_j=1$ is equivalent to $\sum_{j|a^*_j > 0} a^*_j = 1$. These equations imply
\begin{subequations}\label{EqPartA}
\begin{align}
\label{EqPartA-1}& \sum_{j} \frac{\partial  a^*_j}{\partial  q_i} = 0,\\
\label{EqPartA-2}& \sum_{j|a^*_j > 0} \frac{\partial  a^*_j}{\partial  q_i} = 0 .
\end{align}
\end{subequations}
By combining (\ref{eq:marginal1}), (\ref{eq:marginal2}) and \eqref{EqPartA-2}, it follows that
\begin{equation*}
	-\frac{D_i}{\omega} + \frac{n^*}{\omega}\frac{\partial \lambda}{\partial  q_i} = 0,
\end{equation*}
where $n^*$ is the number of agents with strict positive probability of being attacked. By solving this expression for $\partial \lambda/\partial q_i$ and substituting the result in (\ref{eq:marginal1}) and (\ref{eq:marginal2}), we establish the statement.
\end{proof}

{ Proposition~\ref{prop:marginal-changes} bears further consequences for vertex-transitive networks, where each agent obtains the same number of documents in expectation, $D_i=D$. For this reason, more precise results can be obtained, including the following monotonicity property: if an agent invests more in security than another agent, then his attack probability is lower (and vice versa).}
 
\begin{proposition}[Attacks to vertex-transitive networks]\label{prop:attack-vt}
If the network is vertex-transitive then $a^*_i<a^*_j$ if and only if $q_i>q_j$.
\end{proposition}
\begin{proof}
Firstly, we rewrite \eqref{eq:findlambda} to obtain
\begin{align*}
\lambda^*&=\frac\omega{n^*}-\frac{D}{n^*} \sum_{\ell: a^*_\ell>0} (1-q_\ell)\end{align*}
Next if $a^*_i>0$ then
\begin{align*}
a^*_i&=\frac1\omega ((1-q_i) D+ \lambda)\\
&=\frac1{n^*} -\frac{D}{\omega} \big(q_i-\frac1{n^*}\sum_{\ell: a^*_\ell>0} q_\ell\big)
\end{align*}
It is then clear that, provided $a_i^*>0$, $a^*_i<a^*_j$ if and only if $q_i>q_j$.
If instead $a^*_i=0$, then we derive the following equivalent inequalities.
\begin{align*}
(1-q_i) D +\lambda^*& \leq 0\\
\lambda^*& \leq -(1-q_i) D \\
\frac\omega{n^*}-\frac{D}{n^*} \sum_{\ell: a^*_\ell>0} (1-q_\ell)& \leq -(1-q_i) D \\
q_i&\geq \frac\omega{Dn^*} +\frac{1}{n^*} \sum_{\ell: a^*_\ell>0}q_\ell.
\end{align*}
At the same time,  $a^*_j>0$ is equivalent to
\begin{align*}0<&\frac{1}{n^*} -\frac{D}{\omega} \big(q_j-\frac1{n^*}\sum_{\ell: a^*_\ell>0} q_\ell\big)\\
\Leftrightarrow\ q_j&<\frac\omega{Dn^*} +\frac{1}{n^*} \sum_{\ell: a^*_\ell>0}q_\ell.
\end{align*}
Thus, $0=a^*_i<a^*_j$ is equivalent to $q_i>q_j$.
\end{proof}

This result immediately leads to the following implications: (a) maximal investments in security guarantee an upper bound on the attack probability; and (b) if all agents invest the same amount, then the attack vector is uniform.
\begin{corollary} \label{Cor1}
For vertex-transitive networks, there hold true that:
\begin{enumerate}
	\item [(a)] if $q_i=1$ for some $i$, then $a^*_i\leq 1/n$;
\item [(b)] if $q_i=\bar q$ for all $i$, then $a^*_i=1/n$ for all $i$.
\end{enumerate}
\end{corollary}

\subsection{Investments under strategic attacks}

In stage 1 of the security game, the security investments are conveniently modeled as the outcome of a game between the agents. In this game, they take the best response $\mathbf{a}^*(\mathbf{q})$ of the adversary into account. The reward of agent $i$ equals (cf.\ \eqref{eq:ut})
$$\Pi_i = 1 -  \sum_j a^*_j (1-q_j) P_{ij}- \frac12 \alpha q_i^2.$$
%
First we analyse the cooperative case, where the social utility 
$$S =\sum_i \Pi_i = n - \sum_j a^*_j  (1-q_j) D_{j}- \frac12\sum_i \alpha q_i^2$$  
is maximized. The proof of the following result is postponed to the Appendix.
%
\begin{theorem}[Social optimum against strategic attacks]\label{X6proposition3}
In a vertex-transitive network facing a strategic attack, the social optimum $\mathbf q^{O,S}$ is unique and equal to
\begin{align}\label{X6formule3}
	q_i^{O,S} = \frac{D}{\alpha n} \qquad \forall i\in V.
\end{align}
\end{theorem}

\begin{remark}[Uniform investments, uniform attacks and sacrificial lambs]
 
Theorem~\ref{X6proposition3} indicates that it is socially optimal for an agent to invest the same as the others. As a consequence, one may immediately verify that $\mathbf a^*(\mathbf{q}^{O,S})=\frac1n$, that is, the socially optimal uniform investments imply uniform attack probability. In other words, we may say that the socially optimal investments make the strategic advantage of the adversary void.

This recommendation is in contrast with some previous studies, for instance \cite{Bier07} \cite{John12}, suggesting that it might be optimal to leave some agents unprotected and make them {\em sacrificing lambs}. Our setting differs from theirs mainly in the definition of the cost function, which is quadratic in our case (as opposed to linear). We want to stress that the ineffectiveness of sacrificing lambs is not an artefact of our homogeneity assumption. Certainly, the homogeneity of the investments as predicted by Theorem~\ref{X6proposition3} is indeed a consequence of the homogeneity of the networks. However, the indication that {\em optimal investments make the attack probabilities uniform} appear to be valid beyond this scenario. Indeed, detailed  calculations on the star graphs (reported in the Appendix) show that also in that non-homogeneous network an investments strategy that makes the attack probability uniform outperforms a sacrificial lamb strategy.
\end{remark}

Instead, if each agent optimizes his individual reward, the following equilibrium investment levels are attained. 
\begin{theorem}[Equilibrium against strategic attacks]\label{thm:strategy-equilibrium}
In a vertex-transitive network facing a strategic attack, there is a unique { pure strategy} equilibrium vector of investment levels $\mathbf{q}^{N,S}$, which is symmetric and given by
\begin{equation}\label{eq:equilibrium}
q^{N,S}_i = \frac{(n-D)D + \omega}{(n-D)D + \alpha n \omega} \qquad \forall i\in V.
\end{equation} 
\end{theorem}

{  Theorem~\ref{thm:strategy-equilibrium}, whose proof is also postponed to the Appendix, shows that the first stage of the security game results in a unique and symmetric vector of investment levels. Combining this equilibrium with the outcome of the second stage, results in the Stackelberg equilibrium of our game.}
\begin{corollary}[Stackelberg equilibrium]
The security game under strategic attack has a unique Stackelberg equilibrium with investment levels $\mathbf{q}^{N,S}$ and attack vector $\mathbf{a^*}(\mathbf{q}^{N,S})$
given by \eqref{eq:findlambda}, \eqref{eq:finda} and \eqref{eq:equilibrium}.
\end{corollary}

The equilibrium investments in stage 1 are a function of $D$, the expected number of documents obtained, which in turn depends on the transmission probability $p$.

\begin{remark}[Dependence on $p$]
The equilibrium investments~\eqref{eq:equilibrium} are increasing in $p$ for small $p$, till the point where $D=n/2$, after which they are decreasing in $p$.
Indeed,
\begin{align*}
\frac{d }{dp} \big((n-D)D\big)= (n-2D) \frac{dD}{dp}
\end{align*}
and thus
\begin{align}\label{eq:derivativeqtop}
\frac{d q^{N,S}_i}{dp} &= \frac{(n-2D) \frac{dD}{dp} (\alpha n-1)\omega} {\big( (n-D)D + \alpha n \omega \big)^2}
\end{align} 
In view of Proposition~\ref{prop:Dbounds}, the only root of~\eqref{eq:derivativeqtop} is given by $\hat{p}$ such that $D=n/2$. Further, 
$dq^{N,S}_i/dp > 0$ when $D<n/2$ and $dq^{N,S}_i/dp < 0$ when $D>n/2$. 
\end{remark}

\begin{example}[Ring, cont'd]
For ring networks we derive from Example~\ref{exa:ring} that, after neglecting exponential terms, $\hat{p}\simeq 1-\frac4{n+2}$: hence, 
as $n$ diverges, $\hat{p}$ converges to 1. Moreover, 
$$ \lim_{n\to\infty} q^{N,S}_i=\frac{\frac{1+p}{1-p}}{\frac{1+p}{1-p}+\alpha\omega}.$$
This value is strictly larger than the limit social optimum $\lim_{n\to\infty} q^{O,S}_i=0$ as seen in Example~\ref{exa:ring2}.
We conclude that in large rings (which are sparse networks) strategic attacks lead to {\em over-investments}, $q^{N,S}_i>q^{O,S}_i$.
\end{example}

{ 
\begin{example}[Complete, cont'd]
For complete networks we derive from Example~\ref{exa:complete} that, after neglecting exponential terms, $D\simeq n$
as $n$ diverges. This implies that
$$q^{N,S}_i\simeq\frac1{\alpha n}.$$
Comparing this value with the limit social optimum $\lim_{n\to\infty} q^{O,S}_i=\frac1\alpha$, we conclude that in large complete graphs strategic attacks lead to {\em under-investments}.
\end{example}
}

\section{Discussion}\label{sec:discussion}
The investment levels derived in the previous sections can easily be compared. A summary of the most relevant comparisons is given in the following statement.

\begin{theorem}[Investments in vertex-transitive networks]\label{thm:comparison}
Assume the graph $\G$ to be vertex-transitive.
\begin{enumerate}
\item Socially optimal investments do not depend on the type of attack, that is, $q_i^{O,R}=q_i^{O,S}$.
\item { Equilibrium investments are smaller in case of random attacks than in case of strategic attacks, that is, $q_i^{N,R} \le q_i^{N,S}$ and the inequality is strict unless $p=1$. 
\item Random attacks lead to under-investments at equilibrium, that is, $q_i^{N,R}\le q_i^{O,R}$ and the inequality is strict unless $p=0$.}
\item { Strategic attacks can lead to either under- or over-investments. The level of investment depends on the probability $p$: for smaller $p$, over-investments occur, $q_i^{N,S}>q_i^{O,S}$ and for larger $p$, it leads to under-investments occur, $q_i^{N,S}<q_i^{O,S}$. }
Moreover, the condition 
\begin{align} \label{EqStrengthen}
  2(n-D)D \geq (n-2D)(\alpha n - 1)
\end{align}
is sufficient to guarantee a unique transmission probability $p^{*}$ at which the equilibrium investments are socially optimal, $q_i^{N,S}=q_i^{O,S}$. 
\end{enumerate}
\end{theorem}

\begin{proof}
The first three items may be verified immediately by inspection. 
For the fourth item, denote the investments by $q_i(p)$ to stress the dependence on $p$. Observe that that $q_i^{N,S}(1)=q_i^{O,S}(0)=\frac1{\alpha n}$, $q_i^{O,S}(1)=\frac1{\alpha}$ and $q_i^{N,S}(0)>\frac{1}{\alpha n}$. This implies that the graphs of $q_i^{N,S}(p)$ and  $q_i^{O,S}(p)$ intersect at least once. 

Applying the chain rule of differentiation and the fact that $D$ increases with $p$ lead to the following inequalities:
\begin{align*}
\frac{\partial}{\partial p}q_i^{O,S}>&\frac{\partial}{\partial p}q_i^{N,S}\\
 \frac{\partial}{\partial D}\frac{D}{\alpha n}>&\frac{\partial}{\partial D}\left(1- \frac{(\alpha n-1)\omega}{(n-D)D +\alpha \omega n}\right)\\
\frac{1}{\alpha n}>&\frac{(\alpha n-1)\omega(n-2 D)}{\big((n-D)D +\alpha \omega n\big)^2} \\
	(n-D)D +(\alpha\omega n)^2 + &\alpha\omega n \left(  2(n-D)D - (n-2D)(\alpha n-1)\right) > 0.
\end{align*}
A sufficient condition for the latter inequality to be true is given by \eqref{EqStrengthen}.
\end{proof}
{ A few comments about this statement are in order.
Regarding social optima, the reader may find surprising that optimal investments against random and strategic attacks coincide. 
Indeed, the optimal investments against strategic attacks trivialize the strategy of the attacker, that is, make the attack probabilities uniform. This observation is confirmed by the star graph example, described in the Appendix, where the investments that make the attack uniform yield higher reward than sacrificing lamb investments. In the special case of vertex-transitive networks, the uniformity in the attack probability is reflected in the uniformity of the investments.

Regarding equilibria, we observe that equilibrium investments against strategic attacks are always higher than against random attacks, which feature under-investment in comparison with the social optimum. In other words, players that are aware of the strategic nature of the attack shall invest more than against a random attack. This difference is consistent with intuition, since players facing a strategic attacker can be expected to invest more to divert attacks away from themselves. However, the theorem also shows that the awareness of strategic attacks is not sufficient to prevent under-investments  (even though strategic investments remain larger than investments in the random setting). Actually, under-investments against strategic attacks appear precisely when the risk is higher, that is, in the presence of a larger transmission probability and a more tightly connected networks. Indeed, the turning point $p^{*}$ from over- to under-investments is lower in denser networks.

These general facts can be numerically verified in our running examples. We report in Figures~\ref{fig:Investments-complete}, \ref{fig:Investments-ring} and \ref{fig:Investments-star} the optimal and equilibrium investments as functions of the diffusion probability $p$, computed for complete, ring, and star topologies with $n=5$ nodes. Several observations can be made about the similarities and differences between these three very different networks (we remind that the star graph is not vertex-transitive, therefore not covered by Theorem~\ref{thm:comparison}). 
\begin{enumerate}
\item
On all these networks, strategic attacks make the agents invest more at equilibrium than random attacks. 
\item On all these networks, the optimal strategic investments are increasing and become equal to $1/\alpha$ when $p=1$ and equal to $\frac1{\alpha n}$ when $p=0$: at that point they coincide with the equilibrium investments against random attack. Instead, equilibrium investments become equal to $\frac1{\alpha n}$ when $p=1$. The extreme values for $p=0,1$ do not depend on the topology. 
\item On all these networks, the equilibrium results in over-investments for small transmission probabilities $p$ and in under-investments for large transmission probabilities. 
\item For each node, there is a unique probability $p^*_i$ where over-investments pass on to under-investments and equilibrium investments are socially optimal.
\item The transition from over-investments to under-investments takes place at smaller transition probabilities where connectivity is stronger. Consistently, the probability with the largest equilibrium investment level is smaller where connectivity is stronger.
\end{enumerate}

%
%

These observations are consistent with our theoretical results, even if some of them were proved for vertex-transitive networks only, thereby showing that their insights are valid on general networks.
}

%
%
%

\begin{figure}[h]
\centering
  \includegraphics[width=0.6\linewidth]{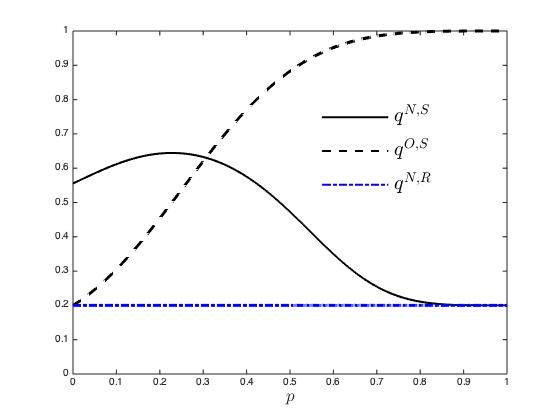}
  \caption{Security investments in a complete graph with $n=5$ nodes where $\alpha=\omega=1$.} 
  \label{fig:Investments-complete}
\end{figure}

\begin{figure}[h]
\centering
  \includegraphics[width=0.6\linewidth]{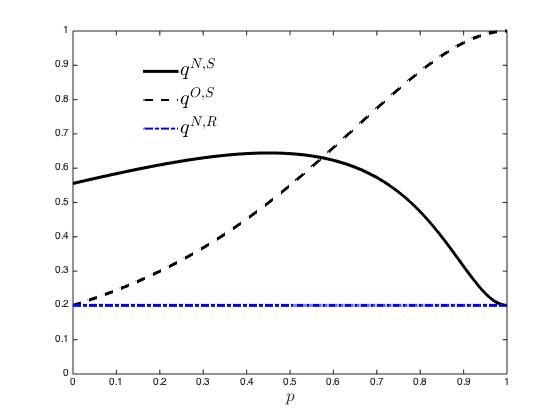}
  \caption{Security investments in a ring graph with $n=5$ nodes where $\alpha=\omega=1$.} 
  \label{fig:Investments-ring}
\end{figure}

\begin{figure}[h]
\centering
  \includegraphics[width=0.6\linewidth]{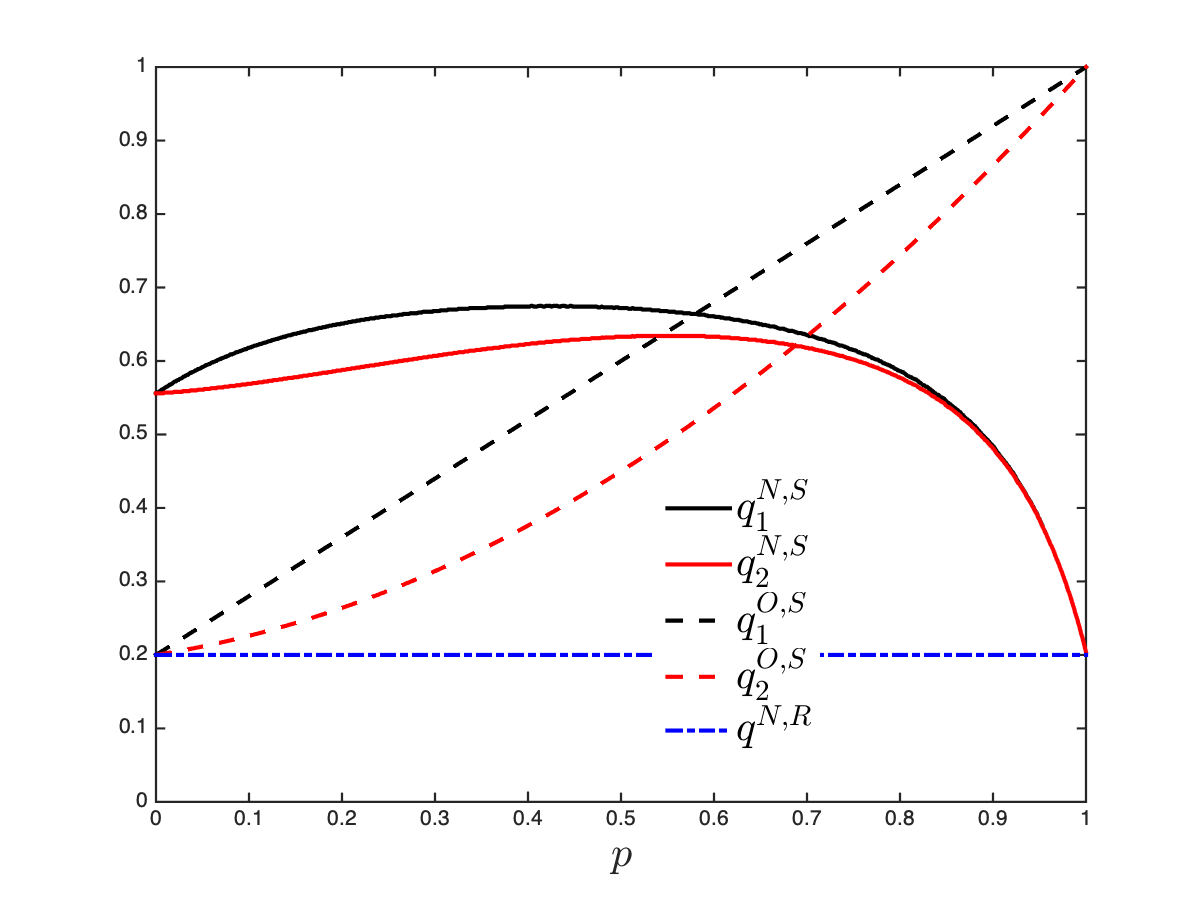}
  \caption{Security investments in a star graph with $n=5$ nodes where $\alpha=\omega=1$.} 
  \label{fig:Investments-star}
\end{figure}

\section{Conclusion}\label{sec:outro}

In this paper, we studied in detail a model of strategic defensive allocation to elucidate the economic forces at play. 
We have shown how the type of attack by the adversary influences the investments by the agents. Equilibrium investments are larger under strategic attacks than under random attacks. Furthermore, in case of random attacks the equilibrium investments are always lower than socially optimal, which represents under-investments in security. Finally, in case of strategic attacks, there are over-investments for small transmission probabilities $p$ and under-investments for large probabilities. This transition takes place at lower probabilities $p$ in more dense networks. { Indeed, those networks where the stakes are higher, because the number of shared documents is larger, are precisely those that are prone to under-investments.}

In a large part of this work, the assumption of vertex-transitivity postulates a homogeneity in the network, which greatly simplifies the analysis. Another  simplifying assumption is the choice of quadratic costs. Even though extending the scope of our analysis would certainly be of interest, we believe that our contribution already exemplifies the fundamental issues of these network privacy games and the key role of the network topology therein.

{ The importance of the network topology is reflected by the fact that optimal investments in random and strategic attacks and equilibrium investments in strategic attacks depend explicitly on the expected number of received documents and, therefore, on the topology. Therefore, the players need to know about the topology to implement their strategies. Since such a knowledge may be hard to obtain in practice, a relevant open question is defining a version of this security game that takes into account suitable limitations of such knowledge.}

\bibliographystyle{plain}
\bibliography{bib4bram}

\appendix

\section{ Information dissemination on the complete graph}

We begin by proving\footnote{The result in Proposition~\ref{Yproposition3} is probably well known. For instance it can be found stated in slide 4 of \url{http://keithbriggs.info/documents/connectivity-Manchester2004Nov19.pdf}. Here we provide a proof for completeness.} formula~\eqref{eq:Qk}.
\begin{proposition}\label{Yproposition3}
Let $Q^n$ be the probability that any document reaches all nodes in $K_n$. Then, for any $p$, it holds that $Q^1 = 1$ and 
\begin{align*}
Q^n=1- \sum_{\ell=1}^{n-1} {{n-1} \choose {\ell-1}} (1-p)^{\ell (n-\ell)} Q^\ell \qquad\forall n>1.
\end{align*}
\end{proposition}
\begin{proof}
Let $\mathcal{T}_{i}^n$ be a transmission network in $K_n$ and observe that $Q^n$ is equal to the probability that $\mathcal{T}_{i}^n$ is connected. Let $C_n(i )$ be the component in which $i$ lies in the transmission network $\mathcal{T}_{i}^n$ and compute
\begin{align*}
\Pr\{\mathcal{T}_{i}^n \text{ is connected }\} &= \Pr\{ |C_n(i)| = n\} \nonumber \\
&= 1- \sum_{\ell=1}^{n-1}  \Pr\{ |C_n(i)| = \ell\},\nonumber
\end{align*}
where $|C_n(i)|$ is the number of nodes in $C_n(i)$. 
To evaluate $\Pr\{ |C_n(i)| = \ell\}$, let $\mathcal{V}_\ell$ be the set of the subsets of $V$ that include node $i$ and have cardinality $\ell$: recognize that there are $\binom{n-1}{\ell-1}$ such subsets. Next, by conditioning on all $\tilde{V} \in \mathcal{V}_\ell$ and exploiting the assumptions of independence between the edges, we can compute
\begin{align}
\Pr\{ |C_n(i)| = \ell\} &= \sum_{ \tilde{V} \in \mathcal{V}_\ell} \Pr\{ C_n(i) = \tilde{V}\}\nonumber \\
& = \sum_{ \tilde{V} \in \mathcal{V}_\ell} \Pr\{ \text{ $\tilde{V}$ is connected in $\mathcal{T}_{i}^n$ } \} \Pr\{ \text { no edge between $\tilde{V}$  and $V \setminus \tilde{V}$ }\} \nonumber \\
& = \sum_{ \tilde{V} \in \mathcal{V}_\ell} \Pr\{ |C_\ell(i)| = \ell\} (1-p)^{\ell(n-\ell)} \nonumber \\ &= \binom{n-1}{\ell-1} (1-p)^{\ell(n-\ell)} \Pr\{ |C_\ell(i) = \ell| \}\label{Yformule3},
\end{align}
so concluding the proof.
\end{proof}
Next, we prove Equation~\eqref{eq:Pk}.
\begin{proposition}\label{Ycorollary4}
In a complete network on $n$ nodes, for every $p$ and all $i \neq j$
\begin{align*}
P_{ij}^n= \sum_{k=2}^{n} {n-2 \choose k-2} (1-p)^{k (n-k)} Q^k.
\end{align*}
\end{proposition}
\begin{proof}
By conditioning on the size of the component in which $j$ lies
\begin{align*}
P^n_{ij} &= \Pr\{ \text{$j$ is connected to $j$ in $\mathcal{T}_i$} \} \nonumber \\
&= \sum_{k=1}^n \Pr\{ \text{$j$ is connected to $j$ in $\mathcal{T}_i$} | ~|C_n(j)| = k \} \Pr\{|C_n(j)| = k\}\\&
= \sum_{k=1}^n \frac{k-1}{n-1}\Pr\{|C_n(j)| = k\} , \nonumber
\end{align*}
where we have used the fact that all nodes are equally likely to be in $C_n(j)$. The result follows by using~\eqref{Yformule3}.
\end{proof}

\section{ Sacrificial vs uniform strategies on the star graph}

Let us first consider the strategy that ensures uniform attack probabilities. From the derivations in Example~\ref{exa:a-star}, we observe that when $\Delta=0$, necessarily $a^*_i=\frac1n$ and $q_1=1-(1-q_2)\frac{D_2}{D_1}$. Therefore,
\begin{align*}
S=&\,n-\sum_j a^*_j (1-q_j) D_j- \frac{\alpha}{2} \sum_j q_j^2\\
=& \,n - (1-q_2) D_2 -\frac{\alpha}2 + \alpha (1-q_2) \frac{D_2}{D_1} - \frac{\alpha}2 (1-q_2)^2 \frac{D_2^2}{D_1^2}-\frac{\alpha}2(n-1)q_2^2
\end{align*}
Its derivative is 
\begin{align*}
\frac{\partial S}{\partial q_2}=& D_2 - \alpha \frac{D_2}{D_1} + \alpha \frac{D_2^2}{D_1^2} - q_2 \big(\alpha \frac{D_2^2}{D_1^2} +\alpha (n-1)\big),
\end{align*}
showing that the reward is optimal for $$ q_2=\displaystyle\frac{\frac{D_2}{\alpha} - \frac{D_2}{D_1} + \frac{D_2^2}{D_1^2}}{ \frac{D_2^2}{D_1^2} + n-1}.$$
Since the expression for the resulting optimal reward is cumbersome, we prefer to present an approximation for large $n$. In that limit, we find 
\begin{align*}
\frac{S^*}{n}=&1-p^2+\frac{p^4}{2\alpha} + O\left(\frac1n\right).
\end{align*}
Moreover, notice that $q^*_2=\frac{p^2}{\alpha}+ O\left(\frac1n\right)$ and $q^*_1=1-p+\frac{p^3}{\alpha}+ O\left(\frac1n\right)$, where the latter quantity is larger than the former: the center has to invest more than the  leaves to ensure uniform attacks.

Let us then compare this reward with that of a sacrificial lamb. In this strategy we assume that one of the leaves is left unprotected, $\subscr{q}{lamb}=0$. In this case, equations~\eqref{eq:find-attack} imply that  $\subscr{a}{lamb}=1$ as long as $q_1$ and $q_2$ are large enough: more precisely, as long as 
\begin{subequations}
\label{eq:so-that-lamb-works}\begin{align}
q_2\ge& \frac{\omega}{D_2}\\
q_1\ge & \frac{\omega + D_1 -D_2}{D_1} 
\end{align}
\end{subequations}
Note that the first quantity goes to zero in the limit of large networks, whereas the second one is approximated by $1-p$: therefore, the lamb strategy is feasible.
Under conditions \eqref{eq:so-that-lamb-works}, we can calculate 
\begin{align*}
\subscr{S}{lamb}=\,&n-\sum_j a^*_j (1-q_j) D_j- \frac{\alpha}{2} \sum_j q_j^2\\
\le\,& n - D_2 -\frac{\alpha}2 \left( \left(1- \frac{D_2}{D_1} +\frac\omega{D_1}\right)^2 + (n-2) \frac{\omega^2}{D_2^2}\right)\\
=\,& n\, (1 - p^2) + o(n).
\end{align*}
This quantity is smaller that $S^*$, thereby showing that the uniform strategy gives higher reward, at least for large enough networks.

\section{Proof of Theorem~\ref{X6proposition3}}
\color{black}
The proof takes four steps. (i) We show that no component of $\mathbf{q}^{O,S}$ is either $0$ or $1$.  
(ii) We deduce the first order conditions (FOC) for optimality of the social optima.
(iii) We show that there is no asymmetric investment level which solves this FOC. 
(iv) We find a symmetric social optimum and prove that this (symmetric) optimum is unique.

(i) Preliminary, we compute the gradient of $S$ as 
$$ \frac{\partial S}{\partial q_i}=-\sum_{j} \frac{\partial a^*_j}{\partial q_i} (1-q_j)D_j   + a^*_iD_i - \alpha q_i.$$
By the assumption of vertex-transitivity this reduces to 
\begin{equation} \label{EqGradS}
	\frac{\partial S}{\partial q_i}=-D \sum_{j}  \frac{\partial a^*_j}{\partial q_i}(1-q_j) + a^*_iD - \alpha q_i.
\end{equation}	

Next, we show that the gradient of $S$, $\nabla(S)$, does not point outward at the boundary of $[0,1]^n$. First,
\begin{align*}
\frac{\partial S}{\partial q_i}(\{ q_i=0,\mathbf{q}_{-i}\}) &=  - D\frac{\partial a^*_i}{\partial  q_i} - D \sum_{j \neq i} \frac{\partial a^*_j}{\partial q_i} (1-q_j) +a^*_iD \nonumber \\
&\geq   - D\frac{\partial  a^*_i}{\partial  q_i} - D \sum_{j \neq i} \frac{\partial a^*_j}{\partial q_i} +a^*_iD \nonumber \\
& = - D \sum_{j} \frac{\partial  a^*_j}{\partial  q_i}+a^*_iD = a^*_i D >0,
\end{align*}
where the final equality follows from \eqref{EqPartA-1}. Second,
\begin{align*}
\frac{\partial S}{\partial q_i}(\{ q_i=1,\mathbf{q}_{-i}\}) &= - \sum_{j \neq i} \frac{\partial a^*_j}{\partial q_i} (1-q_j)D +a^*_i D- \alpha\nonumber \\
&\leq - \sum_{j \neq i} \frac{\partial  a^*_j}{\partial  q_i} (1-q_j) D < 0, \nonumber 
\end{align*}
where the weak inequality follows from $a^*_i D - \alpha \leq 0$ due to $D\leq n$, $a^*_i \leq 1/n$ due to Corollary \ref{Cor1}.(a), and $1\leq \alpha$.

(ii) 
The social optimum $\mathbf{q}^{O,S}$ thus belongs to $(0,1)^n$. From \eqref{EqGradS} and $\partial S/\partial q_i=0$ the social optimum solves for each agent $i$ 
\begin{equation}\label{X6formule2}
\alpha q_i =a^*_i D - D\sum_j \frac{\partial a^*_j}{\partial q_i} (1-q_j).
\end{equation}

(iii) 
In order to prove that all components of $\mathbf{q}^{O,S}$ are equal, without loss of generality let $q_1 = \max{\mathbf{q}^{O,S}}$ and $q_2 = \min{\mathbf{q}^{O,S}}$ and assume that $q_1 > q_2$. We derive a contradiction.
Observe that by~\eqref{X6formule2} 
\begin{align}\label{X6formule2a}
\alpha q_1 &=a^*_1 D - D\frac{\partial a^*_1}{\partial q_1}(1-q_1) - D \sum_{i \neq 1} \frac{\partial a^*_i}{\partial q_1}(1-q_i) \nonumber \\
& =a^*_1 D + D \sum_{i \neq 1} \frac{\partial a^*_i}{\partial q_1}(1-q_1) - D \sum_{i \neq 1} \frac{\partial a^*_i}{\partial q_1}(1-q_i) ,
\end{align}
where the last equality is due to \eqref{EqPartA-1} for $i=1$. Similarly
\begin{align}\label{X6formule2b}
\alpha q_2 &=a^*_2 D - D\frac{\partial a^*_2}{\partial q_2}(1-q_2) - D \sum_{i \neq 2} \frac{\partial a^*_i}{\partial q_2}(1-q_i)\nonumber \\
& =a^*_2 D + D \sum_{i \neq 2} \frac{\partial a^*_i}{\partial q_2}(1-q_2) - D \sum_{i \neq 2} \frac{\partial a^*_i}{\partial q_2}(1-q_i),
\end{align}
with the last equality due to \eqref{EqPartA-1} for $i=2$.
Observe that $a^*_1 < a^*_2$, the definition of $q_1$ implies $0 \leq 1-q_1 \le 1-q^{O,S}_i$ and that $\partial a^*_i/ \partial q_1 \geq 0$ for all $i \neq 1$ by \eqref{eq:derivatives}. Then
\begin{align*}
D \sum_{i \neq 1} \frac{\partial a^*_i}{\partial q_1}(1-q_1) - D \sum_{i \neq 1} \frac{\partial a^*_i}{\partial q_1}(1-q_i) 
= -D \sum_{i \neq 1} \frac{\partial a^*_i}{\partial q_1}(q_1-q_i)
= \alpha q_1 - a^*_1 D< 0,
\end{align*}
with the final equality due to \eqref{X6formule2} and the inequality follows from $q_1-q_\ell>0$ for at least one $\ell$. 
By a similar line of arguments
\begin{align*}
D \sum_{i \neq 2} \frac{\partial a^*_i}{\partial q_2}(1-q_2) - D \sum_{i \neq 2} \frac{\partial a^*_i}{\partial q_2}(1-q_i) 
= \alpha q_2 - a^*_2 D > 0.
\end{align*}
These two inequalities prove that the right-hand side of (\ref{X6formule2b}) is larger than the right-hand side of (\ref{X6formule2a}), which contradicts $q_1>q_2$. Therefore, $q_1=q_2$ and all components of $\mathbf{q}^{O,S}$ are equal.

(iv)
Now we have established that $\mathbf{q}^{O,S}$ is a symmetric social optimal investment level, we  elaborate (\ref{X6formule2}) to derive 
$\alpha q_i^{O,S} =a^*_i D - D (1-q_i^{O,S}) \sum_j\frac{\partial a^*_j}{\partial q_i} = a^*_i D$ 
by \eqref{EqPartA}. By summing $q_i^{O,S} =a^*_i D/\alpha$
over all $i$ and using symmetry we obtain (\ref{X6formule3}).

\section{ Proof of Theorem~\ref{thm:strategy-equilibrium}}
Notice that the agents play a strategic game amongst themselves in stage 1. We refer to the outcome of that stage as an equilibrium.  
The proof is divided into three intermediate steps. \begin{enumerate}
\item We prove that there exists at least one pure strategy equilibrium.
\item 
We prove that the equilibrium is unique and symmetric.
\item We exhibit a symmetric equilibrium.
\end{enumerate}
Let us preliminary recall the reward of agent $i$,
\begin{equation} \label{EqPi}
	\Pi_i=1-\sum_j a^*_j (1-q_j) P_{ij}-\frac12 \alpha q_i^2,
\end{equation}
and that the equilibrium solves $\frac{\partial \Pi_i}{\partial q_i}=0$. The derivative of (\ref{EqPi}) is given by
\begin{align}\label{eq:utdif}
\frac{\partial\Pi_i}{\partial q_i} = a^*_i - \sum_{j \in V} \frac{\partial a^*_j}{\partial q_i}(1-q_j)P_{ij} - \alpha q_i
\end{align}

{\bf Step 1.}
We prove that $\Pi_i$ is quasi-concave in $q_i$. The derivative of \eqref{eq:utdif} is given by
\begin{align}\label{eq:utdiff2}
\frac{\partial^2 \Pi_i}{\partial q_i ^2} &= 2\frac{\partial a^*_i}{\partial q_i}   - \sum_{j \in V} \frac{\partial^2a^*_j}{\partial q_i^2}(1-q_j)P_{ij} - \alpha \nonumber \\
&=-2D \frac{n^*-1}{\omega n^*} - \alpha < 0,
\end{align}
where the second equality follows from~\eqref{eq:derivatives} and $\frac{\partial^2a^*_j}{\partial q_i^2}=0$. As the second derivative of the utility of agent $i$ is negative, we conclude that $\Pi_i$ is actually concave.
We are now in the position to apply the result by Debreu, Fan, Glicksberg \cite{Debr52,Fan52,Glic52}
who showed that a pure strategy Nash equilibrium exists in the strategic form game of stage 1 when the strategy sets are compact and convex, and the utility of each agent is quasi-concave in the agent's own strategy and continuous in the strategy of other agents. 

{\bf Step 2.}
We start by finding the second order derivatives of $\Pi_i$. In (\ref{eq:utdiff2}) we already computed this derivative to $q_i$. Additionally note that the derivative of (\ref{eq:utdif}) to $q_j$ for $j\neq i$ is given by
\begin{align*}
\frac{d^2\Pi_i}{dq_i dq_j} &=\frac{da^*_i}{dq_j} + \frac{da^*_j}{dq_i}P_{ij}  - \sum_{\kappa \in V} \frac{d^2a^*_\kappa}{dq_i dq_j}(1-q_\kappa)P_{i,\kappa} \\
&=\frac{D}{\omega n^*} (1+P_{ij}),
\end{align*}
where $\frac{d^2a^*_\kappa}{dq_i dq_j}=0$ is used in the second equality.

Secondly, we determine the number of agents having a positive probability of being attacked, $n^*$. For any agent $i$
\begin{align}
\frac{\partial \Pi_i}{\partial q_i}(\{0,q_{-i}\}) &= a_i - \sum_{j \neq i} \frac{\partial a_j}{\partial q_i} [1-q_j]P_{i,j} - \frac{\partial a_i}{\partial q_i}  \nonumber \\
&> a_i - \sum_{j \neq i} \frac{\partial a_j}{\partial q_i} - \frac{\partial a_i}{\partial q_i} \nonumber \\
& = a_i - \sum_{j \neq i } \frac{\partial a_j}{\partial q_i} = a_i \geq 0. \label{EqPosQi}
\end{align}
This implies that $q_i>0$: that is, it is not optimal not to investment, since slightly increasing the investment level will result in larger rewards.
Now assume that $a^*_i=0$. By \eqref{EqPi}, the rewards of agent $i$ will be
\[
  \Pi_i=1-\sum_{j\neq i} a^*_j (1-q_j) P_{ij}-\frac12 \alpha q_i^2.
\]
Since the equilibrium investments $q_i$ maximize these rewards, we should have $q_i=0$. But this contradicts our conclusion from \eqref{EqPosQi}. Therefore, our assumption $a^*_i=0$ was false and we must have $a^*_i>0$ for all agents $i$. This implies $n^*=n$, all agents have a positive probability of being attacked.

Combining these results, the negated Jacobian $-J$ with $J_{ij}=\frac{\partial^2\Pi_i}{\partial q_i \partial q_j}$ becomes
\begin{equation}\label{eq:jacobian}
-J=
\begin{bmatrix}
\frac{2n - 2}{\omega n}D + \alpha & -\frac{D}{\omega n}(1+P_{12}) & \cdots & -\frac{D}{\omega n}(1+P_{1n}) \\[2.2ex]
-\frac{D}{\omega n}(1+P_{21}) & \frac{2n - 2}{\omega n}D + \alpha  & \cdots & -\frac{D}{\omega n}(1+P_{2n})  \\[2.2ex]
\vdots & \vdots & \ddots & \vdots \\[2.2ex]
-\frac{D}{\omega n}(1+P_{n1}) & -\frac{D}{\omega n}(1+P_{n2}) & \cdots &\frac{2n - 2}{\omega n}D + \alpha.
\end{bmatrix}
\end{equation}
Next we  show that the  matrix $-J$ is diagonally dominant. 
\begin{align*}
\sum_{j \neq i} |-J_{ij}| &= \sum_{j \neq i}  \frac{D}{\omega n}(1+P_{ij}) \\
&= \frac{D(n-1)}{\omega n} + \frac{D(D-1)}{\omega n}\\
& \leq \frac{D(n-1)}{\omega n} + \frac{D(n-1)}{\omega n}\\
& = D\frac{2n-2}{\omega n} ~\leq~ |-J_{ii}|,\ \mbox{for all\ } i.
\end{align*} 
Because the matrix $-J$ is also symmetric, all principal minors in the negated Jacobian are positive \cite{bapat_raghavan_1997}. Because of this, the Nash equilibrium in a symmetric game is unique \cite{Gale65}. As we already concluded that a pure Nash equilibrium always exists, we are able to conclude that this equilibrium is unique and symmetric. 

{\bf Step 3.} Finally, we exhibit the symmetric equilibrium $\mathbf{q}=q\mathbf{1}$. Because of this symmetry, $a^*_i=1/n$ by Corollary \ref{Cor1}.
Starting from~\eqref{eq:utdif} we obtain \begin{align*}
\frac{\partial\Pi_i}{\partial q_i} 
&= \frac1n- (1-q) \sum_{j \in V} \frac{\partial a^*_j}{\partial q_i}P_{ij} - \alpha q\\
&= \frac1n + (1-q) \frac{n-1}{\omega n} D - (1-q) \frac{D}{\omega n} (D-1) - \alpha q\\
&= \frac1n + (1-q) \frac{D}{\omega n} (n-D) - \alpha q.
\end{align*}
Since the equilibrium solves $\partial\Pi_i/\partial q_i = 0$, 
the expression \eqref{eq:equilibrium} follows immediately.

%
%
%
%
%
%
%
%
%

\end{document}